\documentclass[conference,a4paper]{APSIPA2021}
\usepackage{multirow}
\usepackage{graphicx}
\usepackage{float}
\usepackage{amsmath}
\usepackage[psamsfonts]{amssymb}
\usepackage{amsxtra}
\usepackage{threeparttable}

\usepackage{amsthm}

\usepackage{booktabs}
\usepackage{bm}
\usepackage{makecell, diagbox,array}
\usepackage{hyperref}
\hypersetup{hidelinks}
\usepackage{color}
\newtheorem{prop}{Proposition}

\begin{document}

\title{An MAP Estimation for Between-Class Variance}

\author{%
\authorblockN{%
Jiao Han\authorrefmark{1}\authorrefmark{2},
Yunqi Cai\authorrefmark{1},
Lantian Li\authorrefmark{1},
Guanyu Li\authorrefmark{2},
Dong Wang\authorrefmark{1}
}
\authorblockA{%
\authorrefmark{1}
Center for Speech and Language Technologies, BNRist, Tsinghua University, China \\
E-mail: \{hanjiao,caiyq,lilt\}@cslt.org; wangdong99@mails.tsinghua.edu.cn}
\authorblockA{%
\authorrefmark{2}
Key Laboratory of China’s Ethnic Languages and Information Technology of Ministry of Education, \\ Northwest Minzu University, China \\
E-mail: guanyu-li@163.com}
}

\maketitle
\thispagestyle{empty}

\begin{abstract}
Probabilistic linear discriminant analysis (PLDA) has been widely used in open-set verification tasks, such as speaker verification.
A potential issue of this model is that the training set often contains limited number of classes, which makes the estimation for
the between-class variance unreliable. This unreliable estimation often leads to degraded generalization.
In this paper, we present an MAP estimation for the between-class variance, by employing an Inverse-Wishart prior.
A key problem is that with hierarchical models such as PLDA, the prior is placed on the variance of class means
while the likelihood is based on class members, which makes the posterior inference intractable.
We derive a simple MAP estimation for such a model,
and test it in both PLDA scoring and length normalization.
In both cases, the MAP-based estimation delivers interesting performance improvement.
\end{abstract}

\section{Introduction}
Probabilistic linear discriminant analysis (PLDA)~\cite{Ioffe06,prince2007probabilistic,sizov2014unifying}
has been extensively used in open-set verification tasks, such as speaker verification~\cite{campbell1997speaker,reynolds2002overview,hansen2015speaker}.
It represents the data with a linear Gaussian model, where the between-class distribution is a Gaussian and
the within-class distributions of individual classes are homogeneous Gaussians. The parameters of
this model involve a linear transform matrix $\mathbf{M}$ and the between-class covariance of the data
after the linear transform, and
they can be estimated by maximum likelihood (ML) training.
Once the model has been trained, it is possible to
decide whether two samples are produced from the same class or from two different classes~\cite{Ioffe06},
and this decision is optimal in terms of minimum Bayes risk (MBR)~\cite{wang2020remarks}.

A potential problem of PLDA is the unreliable estimation for the between-class covariance, denoted by $S_B$.
In many applications, the number of classes in the training data is limited. Take speaker verification
as an example, the largest open-source dataset VoxCeleb contains 7,000+ speakers.
Considering the high dimensionality
of the data (e.g., speaker vectors in speaker verification, whose dimension is 400-600), it is difficult to
estimate a reasonable between-speaker covariance with maximum likelihood training.

To have an intuition, we take a simulation experiment by sampling $n$ samples from a Gaussian distribution
and a Laplacian distribution, and then compute the ML-based variance estimation for each sampling.
We plot the variance's variance to show the reliability of the estimation.
As shown in Fig.~\ref{fig:sim},
when the number of samples $n$ is small, the variance's variance is large, indicating that the estimation is highly unreliable.
This conclusion is more clear with the Laplacian distribution, due to its heavy-tail property.

\begin{figure}[!htp]
	\centering
	\includegraphics[width=1.\linewidth]{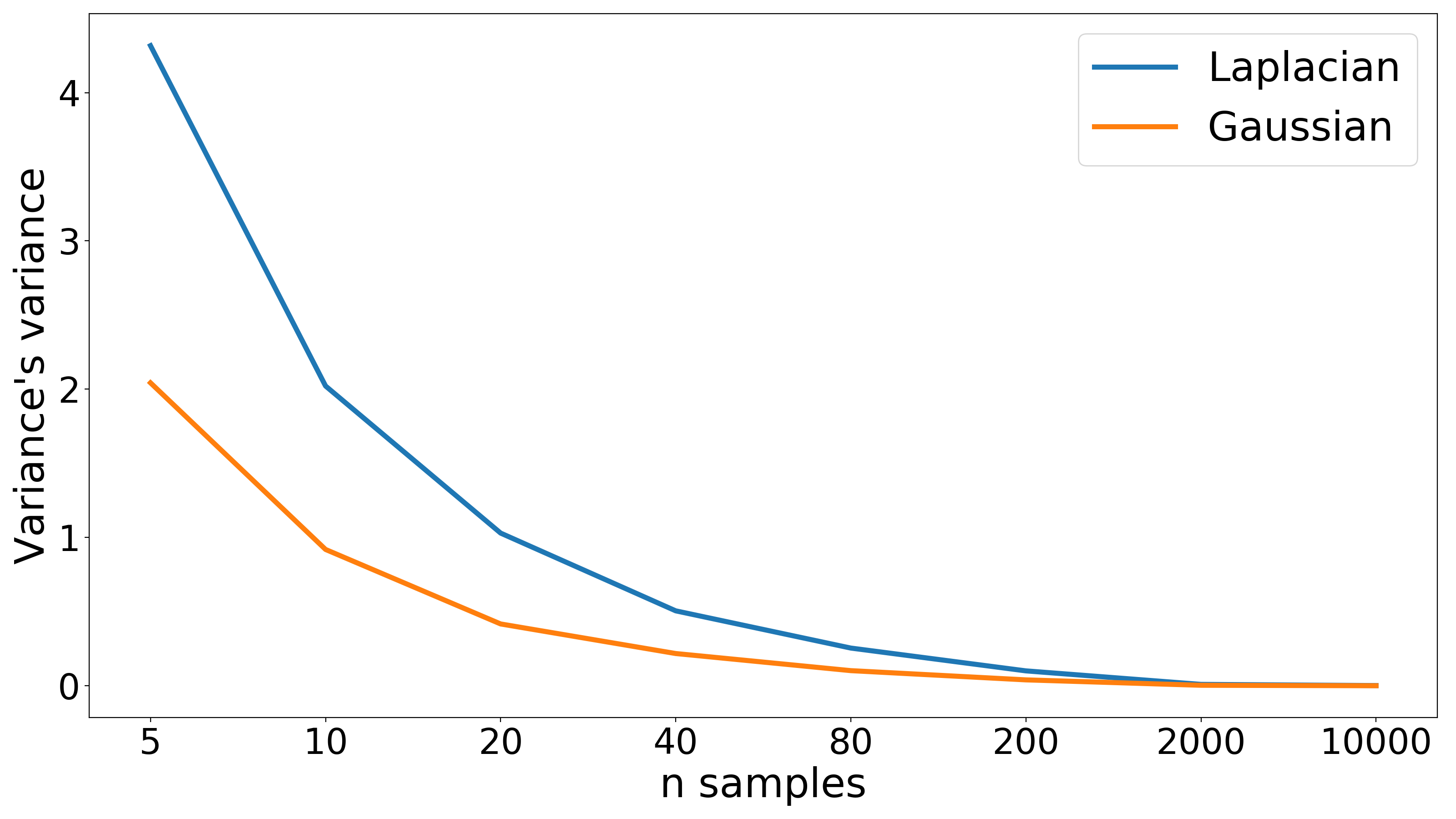}
	\caption{Variance's variance of $n$ samples from Gaussian and Laplacian distributions.
            For each distribution, we firstly sample $n$ data points, and then
            compute their variance $\sigma$ (i.e., ML-based estimation for the underlying true variance).
            This process repeats 10,000 times and the variance of the 10,000 $\sigma$ values is computed.
            The x-axis is the number of samples $n$, and the y-axis is the value of variance's variance.
            For simplicity, the data points are one-dimensional.
            For the sake of comparison, the true variances of two distributions are both set to 1.0.}
	\label{fig:sim}
\end{figure}


For PLDA, the limited number of classes leads to the same unreliable estimation for the between-class covariance $S_B$.
For speaker verification, the distribution of speaker vectors is known to be heavy-tailed~\cite{kenny2010bayesian},
which makes the ML estimation for $S_B$ even more unreliable.
Moreover, the dimensionality of speaker vectors is often as high as 400-600, which further exaggerates
the problem.

In this paper, we propose a robust estimation for the between-class variance $S_B$, by placing an Inverse-Wishart prior on $S_B$ and then conduct
maximum a posterior (MAP) estimation.
At the first glance, the MAP estimation seems trivial if the prior and the associated conditional likelihood are given.
However, for hierarchical probabilistic models such as PLDA, it is much more convolved. This is because the prior is placed
on the covariance of the \emph{class means} while the likelihood is based on the \emph{class members}. This complication
leads to intractable inference for $S_B$.  We will prove that under some mild conditions,
the MAP estimation can be reformulated to a simple linear interpolation of $S_B$ derived by the standard PLDA and a prior covariance.

\section{Theory}

\subsection{Preliminary of PLDA}
We consider the two-covariance form of the PLDA model~\cite{sizov2014unifying}, which assumes a linear Gaussian
 as follows, where $k$ indexes the class, and $i$ indexes samples of a particular class:

\begin{equation}
\pmb{x}_{ki}= \pmb{\mu}_0 + \mathbf{M} \pmb{\mu}_{k} + \pmb{n}_{ki},
\end{equation}
\begin{equation}
\pmb{\mu}_k \sim \mathcal{N} (\mathbf{0}, diag(\pmb{\epsilon})),
\end{equation}
\begin{equation}
\pmb{n}_{ki} \sim \mathcal{N} (\mathbf{0}, \mathbf{I}),
\end{equation}

\noindent where we assume $\mathbf{M}$ is of full rank. By this assumption,
the within-class covariance is the identity matrix $\mathbf{I}$, and the between-class covariance is computed as follows:

\begin{equation}
S_B = \mathbf{M} diag(\pmb{\epsilon}) \mathbf{M}^T.
\end{equation}

The likelihood of the data of a particular class $k$ can be computed as follows:

\begin{eqnarray}
p(\pmb{x}_1,\pmb{x}_2, ..., \pmb{x}_{n_k}) &=& \int p(\pmb{x}_1,...,\pmb{x}_{n_k} | \pmb{\mu}_k) p(\pmb{\mu}_k) {\rm d}\pmb{\mu}_k \nonumber \\
                         &=& \int \prod_{i=1}^{n_k} p(\pmb{x}_i | \pmb{\mu}_k) p(\pmb{\mu}_k) {\rm d}\pmb{\mu}_k. \label{eq:px}
\end{eqnarray}

Since $p(\pmb{x}_i | \pmb{\mu}_k)$ and $p(\pmb{\mu}_k)$ are both Gaussian, it is easy to show that
$p(\pmb{x}_1,\pmb{x}_2, ..., \pmb{x}_{n_k})$ is Gaussian and can be computed efficiently. Collecting all the data of
$K$ classes, the likelihood function can be computed as follows:

\begin{equation}
\label{eq:ml}
\mathcal{L}(\pmb{\epsilon}, \mathbf{M}, \pmb{\mu}_0) = \prod_{k=1}^K p(\pmb{x}^k_1, ..., \pmb{x}^k_{n_k}),
\end{equation}
\noindent where $\pmb{\epsilon}, \mathbf{M}, \pmb{\mu}_0$ are the parameters. Maximizing this function with
respect to these parameters leads to a maximum likelihood (ML) training.

Once the model has been trained, it can be employed to perform verification tasks.
According to the hypothesis test theory~\cite{neyman1933ix},
the following likelihood ratio (LR) is optimal in terms of Bayes risk, when used to judge whether a test sample $\pmb{x}$
belongs to the class represented by the enrollment samples $\{\pmb{x}_1, ..., \pmb{x}_n\}$:

\begin{equation}
\label{eq:lr}
 \emph{LR} = \frac{p(\pmb{x}, \pmb{x}_1, ..., \pmb{x}_n)}{ p(\pmb{x}) p(\pmb{x}_1, ..., \pmb{x}_n)}.
\end{equation}

Thanks to the linear Gaussian form of the model, the LR score has a closed form and can be computed efficiently~\cite{Ioffe06,prince2007probabilistic}.

\subsection{MAP estimation for $S_B$}

Suppose an Inverse-Wishart prior on $\pmb{\epsilon}$~\cite{murphy2007conjugate}:

\begin{equation}
\label{eq:wishart}
p(\pmb{\epsilon}; \pmb{\phi}, \nu) = \frac{1}{Z} (\prod_{j=1}^p \epsilon_j)^{-\frac{\nu+p+1}{2}} \exp \{-\frac{1}{2} \prod_{j=1}^{p} \phi_j \epsilon^{-1}_j \},
\end{equation}
\noindent where $p$ is the dimension of the data and $Z$ is a normalization term.
If there are $n$ observations following a Gaussian,
it is easy to derive that the MAP estimation for $\pmb{\epsilon}$ is given as follows~\cite{murphy2007conjugate}:

\begin{equation}
\label{eq:post}
\pmb{\epsilon}_{\emph{MAP}} = \frac{\pmb{\phi}+ n \pmb{\epsilon}_{\emph{ML}}}{\nu+n+p+1},
\end{equation}

\noindent where $\pmb{\epsilon}_{\emph{ML}}$ denotes the ML estimation for $\pmb{\epsilon}$ with $n$ observations.

If we place an Inverse-Wishart prior on $S_B$ of the PLDA model, the graphical representation is shown in Fig.~\ref{fig:graph}.
In this case, the interaction between $\bm{\epsilon}$ and the observation $\bm{x}$ is indirect, and there are no i.i.d. Gaussian samples
can be used to estimate $\pmb{\epsilon}_{\emph{MAP}}$  by (\ref{eq:post}).
One possibility is to use the likelihood function (\ref{eq:px}) as the conditional probability, and derive the MAP estimation as follows:

\begin{equation}
\pmb{\epsilon}_{\emph{MAP}} = \arg\max_{\pmb{\epsilon}} p(\pmb{\epsilon}; \pmb{\phi}, \nu)  \prod_k p(\pmb{x}_1,\pmb{x}_2, ..., \pmb{x}_{n_k}|\pmb{\epsilon})
\end{equation}

\begin{figure}[!htp]
 \centering
 \includegraphics[width=0.75\linewidth]{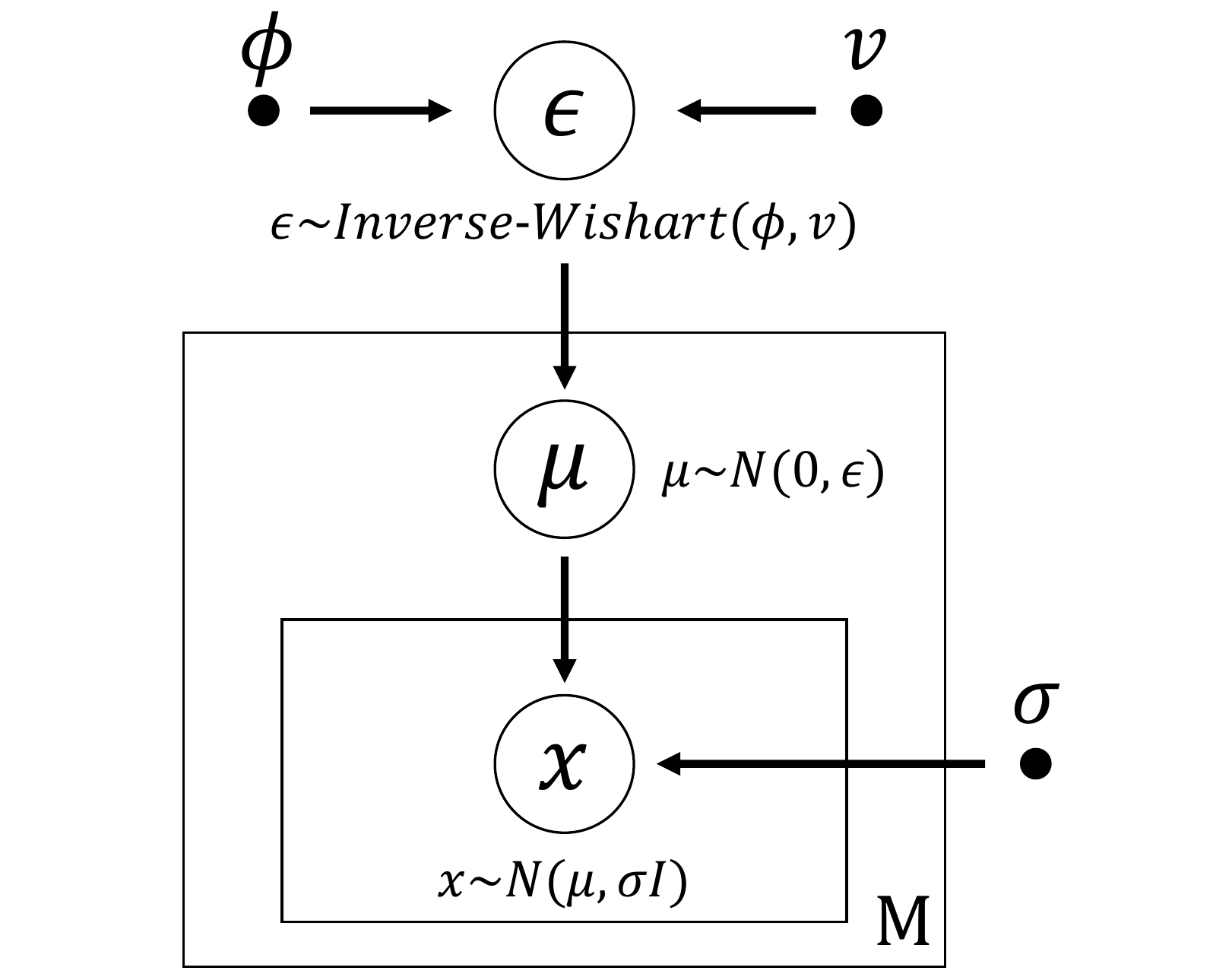}
 \caption{Graphical model of PLDA with an Inverse-Wishart prior on the between-class variance $\pmb{\epsilon}$.}
 \label{fig:graph}
\end{figure}

Since $ \prod_k p(\pmb{x}_1,\pmb{x}_2, ..., \pmb{x}_{n_k}|\pmb{\epsilon})$ involves an undetermined parameter $\mathbf{M}$ and  needs to marginalize on $\pmb{\mu}_k$,
the inference for $\pmb{\epsilon}_{\emph{MAP}}$ is intractable. Although a variational approach can be used~\cite{villalba2011towards}, the iterative
process leads to increased computational load.
We will show that a simple MAP estimation can be derived by using $\pm{\epsilon}$ derived from the standard PLDA, under mild conditions.

\begin{prop}
If every class involves $n$ training samples, and PLDA has been well trained, the between-class covariance
$\bm{\epsilon}$ can be written as $\bm{\epsilon} =  \frac{\sum_k \bar{\bm{x}}'_k}{K} - 1/n$,
where $\bar{\pmb{x}}'_k = \frac{1}{n} \sum_{i=1}^{n} \mathbf{M}^{-1}(\pmb{x}_{ki} - \pmb{\mu}_0)$.
\end{prop}

\begin{proof}

Since the PLDA model has been well trained, $\bm{x}'=\mathbf{M}^{-1}(\bm{x}-\bm{\mu}_0)$ show a regulated distribution: the between-class variance is $diag(\bm{\epsilon})$, and the
within-class variance is $\mathbf{I}$. Considering a particular class, the joint probability of the class members is a function of $\bm{\epsilon}$. Considering a particular dimension $j$,
the probability is given by:

\begin{eqnarray}
&&p(\bm{x}_1,...,\bm{x}_n; \epsilon_j)  \nonumber \\
&\propto& \int \frac{1}{\epsilon_j^{1/2}} \exp \big\{-\frac{1}{2\epsilon_j} \mu^2 \big\} \prod_i^n \exp \big\{-\frac{1}{2}(x'_{ij} - \mu )^2 \big\} \rm{d} \mu \nonumber \\
&=& \frac{1}{\epsilon_j^{1/2}} \int \exp \big\{ -\frac{1}{2\epsilon_j} \mu^2 -\frac{1}{2}\sum_i(x'_{ij} - \mu)^2 \big\} \rm{d} \mu \nonumber \\
&=& \frac{1}{\epsilon_j^{1/2}} \int \exp \big\{ -\frac{(1 + n \epsilon)}{2\epsilon_j} (  (\mu -  \frac{\epsilon_j}{(1 + n \epsilon_j)} \sum_i x'_{ij} )^2 \nonumber \\
&& - \frac{\epsilon_j^2}{(1 + n \epsilon_j)^2} \sum_i (x'_{ij})^2 ) \big\} \rm{d} \mu \nonumber \\
&\propto&\frac{1}{(1+n\epsilon_j)^{1/2}} \exp \big\{\frac{n^2\epsilon_j}{2(1+n\epsilon_j)}(\bar{x}'_j)^2\big\}
\end{eqnarray}

Considering all the classes:

\begin{eqnarray}
\mathcal{L}(\epsilon_j) &=& \sum_{k=1}^K \ln p(x_1^k, x_2^k, ..., x_n^k; \epsilon_j)  \nonumber \\
                        &=& -\frac{K}{2} \ln (n \epsilon_j + 1) +  \frac{n^2\epsilon_j}{2(1+n\epsilon_j)} \sum_k (\bar{x}'_{jk})^2 + const \nonumber
\end{eqnarray}

Take derivative of $\mathcal{L}(\epsilon_j)$ with respect to $\epsilon_j$ and set it to 0:

\begin{equation}
\frac{\partial \mathcal{L} (\epsilon_j)} {\partial \epsilon_j} = -\frac{K}{2} \frac{n}{n\epsilon_j + 1} + \frac{1}{2} \sum_k (\bar{x}'_{jk})^2  \frac{n^2(1+n\epsilon_j) - n^3 \epsilon_j}{(1+n\epsilon_j)^2} = 0
\end{equation}

A simple computation shows:

\begin{equation}
\epsilon_j = \frac{\sum_{k=1}^K (\bar{x}'_{jk})^2}{K} - \frac{1}{n}
\end{equation}

Since all the dimensions of $\bm{x}'$ are independent, we have:

\begin{equation}
\bm{\epsilon} = \frac{\sum_{k=1}^K (\bar{\bm{x}}'_{k})^2}{K} - \frac{1}{n}
\end{equation}

\end{proof}

Obviously, if $n$ is large, the estimation for $\pmb{\epsilon}$ approaches to:

\begin{equation}
\label{eq:eps}
\bm{\epsilon} \approx \frac{\sum_k (\bar{\bm{x}}'_{k})^2}{K},
\end{equation}
\noindent which can be interpreted as an ML estimation for the covariance of a Gaussian distribution represented by the
$K$ \textbf{virtual} samples $\{\bar{\bm{x}}'_k\}$.

Since $\pmb{\epsilon}$ derived by PLDA is equivalent to $\pmb{\epsilon}_{\emph{ML}}$ derived with the $K$ virtual
samples $\{\bar{\bm{x}}'_k\}$, we can use these virtual samples as the observations of the underlying Gaussian
model, and derive the MAP estimation for the covariance of these samples:

\begin{equation}
\label{eq:emap}
\pmb{\epsilon}_{\emph{MAP}} = \frac{\pmb{\phi}+ K \pmb{\epsilon}}{\nu+K+p+1},
\end{equation}
\noindent where $\pmb{\epsilon}$ is obtained from the standard PLDA.

By defining appropriate $\pmb{\phi}$ and $\nu$, the above equation
can be reformulated as a simple linear interpolation:

\begin{equation}
\label{eq:interp}
\pmb{\epsilon}_{\emph{MAP}} = \frac{\alpha \pmb{\epsilon}_0+ K \pmb{\epsilon}}{\alpha + K},
\end{equation}
\noindent where $\pmb{\epsilon}_0$ can be interpreted as a prior covariance, and
$\alpha$ is a hyper-parameter that represents the number of virtual samples associated with $\pmb{\epsilon}_0$.
We therefore derived a simple form of MAP estimation for the between-class variance with PLDA.

We highlight that although the final result (\ref{eq:interp}) looks simple, it should not be regarded as
trivial. In fact,
to derive such a result, we have assumed that all the classes involve the same number of samples. This is even not
true in most practical usage, demonstrating that (\ref{eq:interp}) is not as straightforward as the first glance.
\footnote{Fortunately, one can verify that if each class contains sufficient samples, (\ref{eq:eps}) remains correct
and hence (\ref{eq:interp}) holds.}


\subsection{Applied to length normalization}

The MAP-estimated $\bm{\epsilon}$ (and hence $S_B$) can be used directly in PLDA scoring, which we will call PLDA/MAP.
Moreover, the more robust $S_B$ can be used to improve length normalization (LN) as well.
LN is a simple and effective trick that has been widely used in speaker verification~\cite{garcia2011analysis}.
The key idea is that for a high-dimensional Gaussian distribution,
most of the samples should concentrate on an eclipse surface defined by the covariance.
Suppose the distribution has been aligned to the axes, the eclipse surface will be as follows:

\begin{equation}
\sum_j \frac{x_j^2}{\lambda_j} = p
\end{equation}
\noindent where $\lambda_j$ is the variance of the $j$-th dimension. In the PLDA model,
this variance consists of the between-class variance $\epsilon_j$ and the within-class variance $\sigma = 1$, i.e., $\lambda_j = \epsilon_j + 1 $.
LN scales the speaker vectors to this surface if they are not, with the scale factor computed by:

\begin{equation}
r=\frac{\sqrt p}{\sqrt{\sum_j \frac{x_j^2}{\epsilon_j + 1}}}.
\end{equation}

It has been shown that this scaling can greatly improve the Gaussianality of the speaker vectors, hence making them more suitable for PLDA modeling.

Here we encounter the same problem as in PLDA scoring: if $\pmb{\epsilon}$ is not well estimated, the scaling will be incorrect. In particular for
speaker vectors aligned to the directions corresponding to a large $\epsilon_j$, the scaling tends to be aggressive.
The MAP-based estimation for $\pmb{\epsilon}$ discounts large variance and thus is expected to alleviate this problem.
For that purpose, we simply use $\pmb{\epsilon}_{\emph{MAP}}$ to compute the scale factor $r$.
We will call the revised length normalization as LN/MAP.

\section{Related work}

Brummer et al.~\cite{brummer2010bayesian} presented the initial idea of Bayes PLDA,
with the aim to overcome the shortage associated with the point estimation for the
parameters in the conventional ML-based PLDA model.

Villalba et al.~\cite{villalba2011towards} employed a Bayes approach to
improve speaker verification with i-vectors~\cite{dehak2011front}.
This approach was further extended to deal with domain adaptation, where the PLDA parameters
obtained in one domain were used as the prior when training PLDA in a new
domain~\cite{villalba2012bayesian,ito2008speaker}.
Although theoretically interesting, it relies variational inference in both training and test,
which is not very friendly and so is rarely used.


The Inverse-Wishart distribution was generally used as the prior for distance/correlation matrix.
For example, Fang et al.~\cite{fang2013bayesian} employed this prior to regularize the
metric learning with an i-vector system. Ito et al.~\cite{ito2008speaker} employed this prior to adapt
the covariance in the GMM-UBM architecture for speaker verification.



\section{Experiments}
\label{sec:exp}

We evaluate the proposed approach by a speaker verification task, following the deep speaker embedding framework~\cite{deng2014deep,ehsan14,li2017deep}.
Given a speech segment, a speaker vector is produced by a deep neural network which consists of frame-level
feature extractor and utterance-level pooling.
In this paper, we employ the x-vector model~\cite{snyder2018xvector} to produce the speaker vectors.
This model is trained using the Kaldi toolkit~\cite{povey2011kaldi}, following the SITW recipe
\footnote{https://github.com/kaldi-asr/kaldi/tree/master/egs/sitw/}.
The dimensionality of the x-vectors was set to $512$. Once the speaker vectors are obtained,
a PLDA model with LDA dimension reduction is trained and employed to score the test trials.
Note that our research goal here is to demonstrate the MAP-based estimation for $S_B$ rather than present
a SOTA speaker verification system. For this purpose, using a public recipe in Kaldi is a reasonable choice.
Readers can refer to~\cite{matvejka2019analysis,VILLALBA2020101026} for SOTA performance on the same task.

\begin{table*}[!htb]
 \caption{EER(\%) results with different settings of PLDA and LN}
 \label{tab:res}
 \centering
   \scalebox{1.0}{
 \begin{tabular}{lllcccc}
  \cmidrule(r){1-7}
                             &  No.  &  Model         &  SITW.Dev   &  SITW.Eval        &  HI-MIA.Dev   &  HI-MIA.Eval       \\
  \cmidrule(r){1-7}
  \multirow{6}{*}{LDA[512]}  & 1     & PLDA           &  3.697      &  4.019             &  1.080         &  0.891             \\
                             & 2     & PLDA/MAP       & \textbf{3.466}  & \textbf{3.909}   & \textbf{0.945} &  \textbf{ 0.810 }  \\
  \cmidrule(r){2-7}
                             & 3     & PLDA + LN          &  4.005  &  4.647                       & 1.484          & 1.296               \\
                             & 4     & PLDA + LN/MAP      &  3.928  &  4.483                       &  1.350         &  1.134              \\
                             & 5     & PLDA/MAP + LN      &  3.889  &  4.429                       & 1.080          & 0.891               \\
                             & 6     & PLDA/MAP + LN/MAP  &  3.812  &  4.374                       &  1.080         & 0.891               \\
  \cmidrule(r){1-7}
  \multirow{6}{*}{LDA[150]}  & 7     & PLDA           &  3.273      &  3.800                       & 1.215          &  0.972             \\
                             & 8     & PLDA/MAP       &  3.196      &  3.745                       & \textbf{1.080} &  \textbf{0.891}    \\
  \cmidrule(r){2-7}
                             & 9     & PLDA + LN      &  2.965      &  3.362                       &  1.350         &  1.134             \\
                             & 10    & PLDA + LN/MAP  &  3.003      & \textbf{3.335}               &  1.350         &  1.053             \\
                             & 11    & PLDA/MAP + LN  &  3.003      &  3.417                       &  1.215         &  1.134             \\
                             & 12    & PLDA/MAP + LN/MAP  & \textbf{2.926}  & 3.417                & 1.080          &  0.972             \\
  \cmidrule(r){1-7}
  \cmidrule(r){1-7}
  \end{tabular}}
\end{table*}

\subsection{Data}

Three datasets were used in our experiments: VoxCeleb, SITW, and HI-MIA. Details are as follows:

\textbf{VoxCeleb}~\cite{nagrani2017voxceleb,chung2018voxceleb2}: An open-source speaker dataset collected from media sources by University of Oxford.
This dataset contains 2,000+ hours of speech signals from 7,000+ speakers.
This dataset was used to train the x-vector model and the PLDA model used in the test on the SITW dataset.

\textbf{SITW}~\cite{mclaren2016speakers}: A standard evaluation dataset consists of 299 speakers.
The core-core trials built on the SITW.Dev set was used to optimize the prior weight $\alpha$ in the MAP estimation of (\ref{eq:interp}).
The core-core trials built on the SITW.Eval set were used for evaluation.


\textbf{HI-MIA}~\cite{qin2020hi}: An open-source text-dependent speaker recognition dataset.
All the speech utterances contain the word `Hi MIA', recorded by a microphone 3 meters away from the speaker.
The development set (used for training PLDA and estimating the MAP prior weight) involves 5,062 utterances from 254 speakers, and the evaluation
set involves 1,665 utterances from 86 speakers.

\subsection{Behavior of the MAP estimation}

In the first experiment, we study the behavior of MAP estimation using the SITW.Dev core-core trials.
We set the prior covariance $\pmb{\epsilon}_0$ = 1 in ~(\ref{eq:interp}),
and set $K$ as the number of speakers used for training the PLDA model,
which is 6,300 in our experiment.
The performance of the PLDA/MAP on SITW.Dev in terms of equal error rate (EER) is reported in Fig.~\ref{fig:sitw-dev},
where the prior weight $\alpha$ changes from 0 to 7,000.
Note that PLDA/MAP with $\alpha$ = 0 is just the conventional PLDA.
It is clear to see that PLDA/MAP can substantially improve system performance with an appropriate $\alpha$.
Notice that there is an optimal $\alpha$ that best trades off the contribution of the prior and the data.

\begin{figure}[!htbp]
\centering
\includegraphics[width=1.\linewidth]{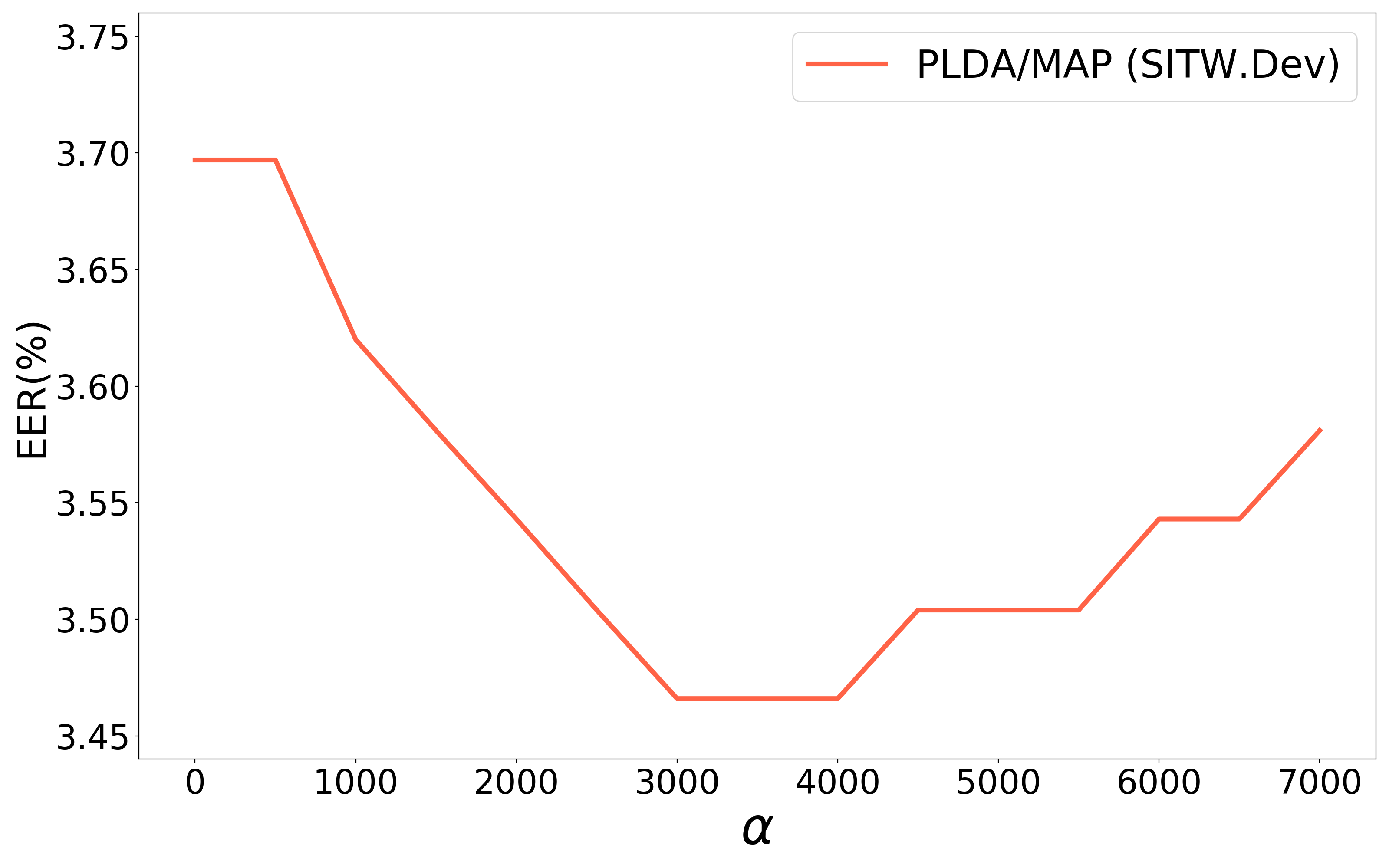}
\caption{EER results of PLDA/MAP on SITW.Dev with different $\alpha$.}
\label{fig:sitw-dev}
\end{figure}

\subsection{Detailed results}

In this experiment, we choose $\alpha$ using the development sets (SITW.Dev for SITW.Eval test, and HI-MIA.Dev for HI-MIA.Eval test)
based on the EER results with PLDA/MAP, and then apply the optimal $\alpha$ to both PLDA/MAP and LN/MAP.
The EER results are reported in Table~\ref{tab:res}.

Firstly, we observe that in almost all the cases, PLDA/MAP outperforms the conventional PLDA.
The general improvement obtained by PLDA/MAP implies that the MAP estimation indeed delivers a better between-class covariance.

Secondly, we found that in most tests, LN/MAP clearly outperforms the standard LN.
This double confirms that the MAP estimation produces a better between-class covariance.
Moreover, since the improvement was obtained by using the prior weight $\alpha$ selected based on PLDA/MAP,
we conclude that the priors for PLDA/MAP and LN/MAP are consistent.

Thirdly, we observe that in the LDA[512]-dim tests, PLDA/MAP + LN (system 5) generally outperforms PLDA + LN (system 3),
but this is not always the case in the LDA[150]-dim tests.
For example, in the SITW test, PLDA/MAP + LN (system 11) performs worse than PLDA + LN (system 9).
A possible reason is that with the LN operation, the data statistics has been changed,
and so the MAP estimation based on the original statistics may be suboptimal.
In general, LN/MAP is more safe than PLDA/MAP as it can improve performance in
almost all the cases, though in many cases, PLDA/MAP can deliver more improvement than LN/MAP.

Finally, it seems that combining PLDA/MAP and LN/MAP (system 6) may lead to performance improvement in some cases, but this is not always the case.
The improvement, even if it is observed, is not significant. This can be explained again by the suboptimum of the MAP estimation
for the data after length normalization.

\section{Conclusions}

We presented a simple form of MAP estimation for the between-class covariance in the PLDA model.
Our derivation shows that under mild conditions, the MAP estimation can be formed as a linear interpolation
of the ML estimation obtained by standard PLDA and a prior covariance.
We employed the MAP-estimated between-class covariance
to both PLDA scoring and length normalization, and interesting performance improvement was obtained.
Future work will investigate better strategies to combine MAP estimation and length normalization.

\section*{Acknowledgment}

This work was supported by the National Natural Science Foundation of China (NSFC) under Grants No.61633013 and No.62171250, 
and also the Innovation Research Program of Northwest Minzu University under Grant No.YXM2021005.

\bibliographystyle{IEEEtran}
\bibliography{ref}

\begin{thebibliography}{10}
\providecommand{\url}[1]{#1}
\csname url@samestyle\endcsname
\providecommand{\newblock}{\relax}
\providecommand{\bibinfo}[2]{#2}
\providecommand{\BIBentrySTDinterwordspacing}{\spaceskip=0pt\relax}
\providecommand{\BIBentryALTinterwordstretchfactor}{4}
\providecommand{\BIBentryALTinterwordspacing}{\spaceskip=\fontdimen2\font plus
\BIBentryALTinterwordstretchfactor\fontdimen3\font minus
  \fontdimen4\font\relax}
\providecommand{\BIBforeignlanguage}[2]{{%
\expandafter\ifx\csname l@#1\endcsname\relax
\typeout{** WARNING: IEEEtran.bst: No hyphenation pattern has been}%
\typeout{** loaded for the language `#1'. Using the pattern for}%
\typeout{** the default language instead.}%
\else
\language=\csname l@#1\endcsname
\fi
#2}}
\providecommand{\BIBdecl}{\relax}
\BIBdecl

\bibitem{Ioffe06}
S.~Ioffe, ``Probabilistic linear discriminant analysis,'' in \emph{European
  Conference on Computer Vision (ECCV)}.\hskip 1em plus 0.5em minus 0.4em\relax
  Springer, 2006, pp. 531--542.

\bibitem{prince2007probabilistic}
S.~J. Prince and J.~H. Elder, ``Probabilistic linear discriminant analysis for
  inferences about identity,'' in \emph{2007 IEEE 11th International Conference
  on Computer Vision}.\hskip 1em plus 0.5em minus 0.4em\relax IEEE, 2007, pp.
  1--8.

\bibitem{sizov2014unifying}
A.~Sizov, K.~A. Lee, and T.~Kinnunen, ``Unifying probabilistic linear
  discriminant analysis variants in biometric authentication,'' in \emph{Joint
  IAPR International Workshops on Statistical Techniques in Pattern Recognition
  (SPR) and Structural and Syntactic Pattern Recognition (SSPR)}.\hskip 1em
  plus 0.5em minus 0.4em\relax Springer, 2014, pp. 464--475.

\bibitem{campbell1997speaker}
J.~P. Campbell, ``Speaker recognition: A tutorial,'' \emph{Proceedings of the
  IEEE}, vol.~85, no.~9, pp. 1437--1462, 1997.

\bibitem{reynolds2002overview}
D.~A. Reynolds, ``An overview of automatic speaker recognition technology,'' in
  \emph{IEEE International Conference on Acoustics, Speech and Signal
  Processing (ICASSP)}, vol.~4.\hskip 1em plus 0.5em minus 0.4em\relax IEEE,
  2002, pp. IV--4072.

\bibitem{hansen2015speaker}
J.~H. Hansen and T.~Hasan, ``Speaker recognition by machines and humans: A
  tutorial review,'' \emph{IEEE Signal processing magazine}, vol.~32, no.~6,
  pp. 74--99, 2015.

\bibitem{wang2020remarks}
D.~Wang, ``Remarks on optimal scores for speaker recognition,'' \emph{arXiv
  preprint arXiv:2010.04862}, 2020.

\bibitem{kenny2010bayesian}
P.~Kenny, ``Bayesian speaker verification with heavy-tailed priors.'' in
  \emph{Proceedings of Odyssey: The Speaker and Language Recognition Workshop},
  2010, p.~14.

\bibitem{neyman1933ix}
J.~Neyman and E.~S. Pearson, ``Ix. on the problem of the most efficient tests
  of statistical hypotheses,'' \emph{Philosophical Transactions of the Royal
  Society of London. Series A, Containing Papers of a Mathematical or Physical
  Character}, vol. 231, no. 694-706, pp. 289--337, 1933.

\bibitem{murphy2007conjugate}
K.~P. Murphy, ``Conjugate bayesian analysis of the gaussian distribution,''
  \emph{def}, vol.~1, no. 2$\sigma$2, p.~16, 2007.

\bibitem{villalba2011towards}
J.~Villalba and N.~Br{\"u}mmer, ``Towards fully bayesian speaker recognition:
  Integrating out the between-speaker covariance,'' in \emph{Twelfth Annual
  Conference of the International Speech Communication Association}, 2011.

\bibitem{garcia2011analysis}
D.~Garcia-Romero and C.~Y. Espy-Wilson, ``Analysis of i-vector length
  normalization in speaker recognition systems,'' in \emph{Proceedings of the
  Annual Conference of International Speech Communication Association
  (INTERSPEECH)}, 2011.

\bibitem{brummer2010bayesian}
N.~Br{\"u}mmer, ``Bayesian plda,'' Tech. Rep., Agnitio Labs, Tech. Rep., 2010.

\bibitem{dehak2011front}
N.~Dehak, P.~J. Kenny, R.~Dehak, P.~Dumouchel, and P.~Ouellet, ``Front-end
  factor analysis for speaker verification,'' \emph{IEEE Transactions on Audio,
  Speech, and Language Processing}, vol.~19, no.~4, pp. 788--798, 2011.

\bibitem{villalba2012bayesian}
J.~Villalba and E.~Lleida, ``Bayesian adaptation of {PLDA} based speaker
  recognition to domains with scarce development data,'' in \emph{Proceedings
  of Odyssey: The Speaker and Language Recognition Workshop}, 2012, pp. 47--54.

\bibitem{ito2008speaker}
T.~Ito, K.~Hashimoto, Y.~Nankaku, A.~Lee, and K.~Tokuda, ``Speaker recognition
  based on variational bayesian method,'' in \emph{Ninth Annual Conference of
  the International Speech Communication Association}, 2008.

\bibitem{fang2013bayesian}
X.~Fang \emph{et~al.}, ``Bayesian distance metric learning on i-vector for
  speaker verification,'' Ph.D. dissertation, Massachusetts Institute of
  Technology, 2013.

\bibitem{deng2014deep}
L.~Deng and D.~Yu, ``Deep learning: methods and applications,''
  \emph{Foundations and trends in signal processing}, vol.~7, no. 3--4, pp.
  197--387, 2014.

\bibitem{ehsan14}
E.~Variani, X.~Lei, E.~McDermott, I.~L. Moreno, and J.~Gonzalez-Dominguez,
  ``Deep neural networks for small footprint text-dependent speaker
  verification,'' in \emph{IEEE International Conference on Acoustics, Speech
  and Signal Processing (ICASSP)}, 2014, pp. 4052--4056.

\bibitem{li2017deep}
L.~Li, Y.~Chen, Y.~Shi, Z.~Tang, and D.~Wang, ``Deep speaker feature learning
  for text-independent speaker verification,'' in \emph{Proceedings of the
  Annual Conference of International Speech Communication Association
  (INTERSPEECH)}, 2017, pp. 1542--1546.

\bibitem{snyder2018xvector}
D.~Snyder, D.~Garcia-Romero, G.~Sell, D.~Povey, and S.~Khudanpur, ``X-vectors:
  Robust {DNN} embeddings for speaker recognition,'' in \emph{IEEE
  International Conference on Acoustics, Speech and Signal Processing
  (ICASSP)}.\hskip 1em plus 0.5em minus 0.4em\relax IEEE, 2018, pp. 5329--5333.

\bibitem{povey2011kaldi}
D.~Povey, A.~Ghoshal, G.~Boulianne, L.~Burget, O.~Glembek, N.~Goel,
  M.~Hannemann, P.~Motlicek, Y.~Qian, P.~Schwarz \emph{et~al.}, ``The {Kaldi}
  speech recognition toolkit,'' in \emph{IEEE workshop on automatic speech
  recognition and understanding}, 2011.

\bibitem{matvejka2019analysis}
P.~Mat{\v{e}}jka, O.~Plchot, H.~Zeinali, L.~Mo{\v{s}}ner, A.~Silnova,
  L.~Burget, O.~Novotn{\`y}, and O.~Glembek, ``Analysis of but submission in
  far-field scenarios of voices 2019 challenge,'' in \emph{Proceedings of the
  Interspeech. 2019}, 2019, pp. 15--19.

\bibitem{VILLALBA2020101026}
J.~Villalba, N.~Chen, D.~Snyder, D.~Garcia-Romero, A.~McCree, G.~Sell,
  J.~Borgstrom, L.~P. García-Perera, F.~Richardson, R.~Dehak, P.~A.
  Torres-Carrasquillo, and N.~Dehak, ``State-of-the-art speaker recognition
  with neural network embeddings in {NIST} {SRE18} and speakers in the wild
  evaluations,'' \emph{Computer Speech \& Language}, vol.~60, p. 101026, 2020.

\bibitem{nagrani2017voxceleb}
A.~Nagrani, J.~S. Chung, and A.~Zisserman, ``{VoxCeleb}: a large-scale speaker
  identification dataset,'' in \emph{Proceedings of the Annual Conference of
  International Speech Communication Association (INTERSPEECH)}, 2017.

\bibitem{chung2018voxceleb2}
J.~S. Chung, A.~Nagrani, and A.~Zisserman, ``{VoxCeleb2}: Deep speaker
  recognition,'' in \emph{Proceedings of the Annual Conference of International
  Speech Communication Association (INTERSPEECH)}, 2018, pp. 1086--1090.

\bibitem{mclaren2016speakers}
M.~McLaren, L.~Ferrer, D.~Castan, and A.~Lawson, ``The speakers in the wild
  (sitw) speaker recognition database.'' in \emph{Proceedings of the Annual
  Conference of International Speech Communication Association (INTERSPEECH)},
  2016, pp. 818--822.

\bibitem{qin2020hi}
X.~Qin, H.~Bu, and M.~Li, ``{HI-MIA}: A far-field text-dependent speaker
  verification database and the baselines,'' in \emph{IEEE International
  Conference on Acoustics, Speech and Signal Processing (ICASSP)}.\hskip 1em
  plus 0.5em minus 0.4em\relax IEEE, 2020, pp. 7609--7613.

\end{thebibliography}

\end{document}